\documentclass[12pt]{amsart}
\usepackage{amssymb,latexsym}
\usepackage{amsfonts}
\usepackage{amsmath}
\usepackage[colorlinks,linkcolor=blue,anchorcolor=blue,citecolor=blue]{hyperref}
\usepackage{algorithm}
\usepackage{enumerate}
\usepackage{algpseudocode}
\usepackage{verbatim}
\usepackage{graphicx}

\def\F{\mathbb{F}}

\def\Tr{\text{\rm Tr}}

\def\Ker{\text{\rm Ker}}

\def\D{\textrm{D}}

\newtheorem{theorem}{Theorem}[section]
\newtheorem{lemma}[theorem]{Lemma}

\theoremstyle{definition}

\newtheorem{remark}[theorem]{Remark}

\numberwithin{equation}{section}

\begin{document}

\title[Weight hierarchies]{Weight hierarchies of 3-weight linear codes from two $p$-ary quadratic functions }

\author{Xiumei Li}
\address{School of Mathematical Sciences, Qufu Normal University, Qufu Shandong, {\rm 273165}, China}
\email{lxiumei2013@qfnu.edu.cn}

\author{Fei Li}
\address{Faculty of School of Statistics and Applied Mathematics,
Anhui University of Finance and Economics, Bengbu, {\rm 233030}, Anhui, P.R.China}
\email{cczxlf@163.com}



\begin{abstract}
The weight hierarchy of a linear code has been an important research topic in coding theory since Wei's original work in 1991.
Choosing $ D=\Big\{(x,y)\in \Big(\F_{p^{s_1}}\times\F_{p^{s_2}}\Big)\Big\backslash\{(0,0)\}: f(x)+g(y)=0\Big\}$ as a defining set ,
where $f(x),g(y)$ are quadratic forms over $\mathbb{F}_{p^{s_i}},i=1,2$, respectively, with values in $\F_p$,
we construct a family of 3-weight $p$-ary linear codes and
determine their weight distributions and weight hierarchies completely.
Most of the codes can be used in secret sharing schemes.
\end{abstract}

\maketitle

\section{Introduction}

\label{intro}


Let $p$ be an odd prime number and $ \mathbb{F}_{p^s} $ be the finite field with $ p^{s} $
elements. Denote by $\mathbb{F}_{p^s}^{*}$ the set of the nonzero elements of $ \mathbb{F}_{p^s} $.

Let $C$ be a $k$-dimensional subspace of $ \mathbb{F}_{p}^{n} $. It is called an $[n,k,d]\ p$-ary linear code of length $n$ with  minimum (Hamming) distance $d$.
For $ 1\leq r\leq k$, the concept of generalized Hamming weight (GHW) $d_r(C)$ of $C$ can be viewed as an extension of Hamming weight.
We recall the definition of GHWs of linear codes.
Let
$ [C,r]_{p} $ be the set of all the $\mathbb{F}_{p}$-vector subspaces with dimension $r$.
For $ H \in [C,r]_{p}$, the support $ \textrm{Supp}(H)$ of $H$ is the set of coordinates where not all codewords of $H$ are zero, that is,
$$ \textrm{Supp}(H)=\Big\{i:1\leq i\leq n, c_i\neq 0 \ \ \textrm{for some $c=(c_{1}, c_{2}, \cdots , c_{n})\in H$}\Big\}.$$
The $r$-th generalized Hamming weight of $C$, which is also called the $r$-th minimum support weight, is
$$
d_{r}(C)=\min\Big\{|\textrm{Supp}(H)|:H\in [C,r]_{p}\Big\}, \ 1\leq r\leq k.
$$
It is easy to see that $ d=d_{1}(C)$.
The sequence $ \{d_{1}(C),d_{2}(C),\cdots,d_{k}(C)\}$ is called the weight hierarchy of $C$.

The generalized Hamming weight was introduced in 1977 by Helleseth, Kl{\o}ve et al. \cite{9HK92,K78}. Victor Wei \cite{20WJ91} proved that they can characterize the cryptography performance of a linear code over the wiretap channel of type II. From then on, much more attention was paid to generalized Hamming weights. The weight hierarchies for some well-known classes of codes were determined, such as Hamming codes, Reed-Muller codes, Reed-Solomon codes and Golay codes.  To determine the weight hierarchy of linear codes is relatively challenging. During the past three decades, there were research results
about weight hierarchies of some classes of linear codes \cite{11HP98,14JL97,20WJ91,27WZ94,13JF17,18LF17,19LF17,LL20-0,LL20,LL22,19LW19,21XL16,22YL15,HLL23}.

A generic construction of linear codes was proposed by Ding et al. \cite{DD14,DD15,DLN08, DN07}  as below.
Denote by $\Tr^s$ the trace function from $\mathbb{F}_{p^s}$ onto $\mathbb{F}_p$.
Let $ D= \{d_{1},d_{2},\cdots,d_{n}\}$ be a subset of $\mathbb{F}_{p^s}^{\ast}.$
A $p$-ary linear code of length $n$ is defined as follows:
\begin{eqnarray}\label{defcode0}
         C_{D}=\{\left( \Tr^s(xd_1), \Tr^s(xd_2),\ldots, \Tr^s(xd_{n})\right):x\in \mathbb{F}_{p^{s}}\},
\end{eqnarray}
and $D$ is called the defining set of $C_{D}$.
Li et al. \cite{LYF18} extended Ding's defining-set construction, they constructed a $p$-ary linear codes $C_{D}$ by
\begin{eqnarray}\label{defcode2}
         C_{D}=\{( \Tr^s(ax+by))_{(x,y)\in D}:a,b\in \mathbb{F}_{p^{s}}\},
\end{eqnarray}
where $D$ is a subset of $\F_{p^s}^2\backslash\{(0,0)\}$. Furthermore, many linear
codes have been obtained from some cryptographic functions in above two construction method, some optimal linear
codes with a few weights can be constructed (see \cite{5DJ16,DD14,DD15,DLN08, DN07,19TXF17,LYF18,JLF19,LL20,LL22,HLL23,S22} and the references therein).

In this paper, we fix the following notation. Let $s_i$ be positive integers, $q_i= p^{s_i}, i=1,2$. Denote $\mathbb{F}=\mathbb{F}_{q_1}\times\mathbb{F}_{q_2},\mathbb{F}^\star=\mathbb{F}\Big\backslash\{(0,0)\}$ and $s=s_1+s_2$.
Define
a $p$-ary linear code $ C_{D} $ as follows:
\begin{eqnarray}\label{defcode1}
         C_{D}=\Big\{\Big(\Tr^{s_1} (x x_i)+\Tr^{s_2} (y y_i)\Big)_{(x_i,y_i)\in D}:(x,y)
         \in \mathbb{F}\Big\},
\end{eqnarray}
where \begin{align}\label{set:D1}
D=\Big\{(x,y)\in \mathbb{F}^\star: f(x)+g(y)=0\Big\},
\end{align}
where $f(x),g(y)$ are quadratic forms over $\mathbb{F}_{q_i},i=1,2$, respectively, with values in $\F_p$. We will study the weight distribution and weight hierarchy of $C_{D}$ in \eqref{defcode1}.

For the linear codes $C_D$ defined in \eqref{defcode1}, a general formula is likely be employed to calculate the generalized Hamming weight $d_r(C_D)$. It is presented in the following lemma.

\begin{lemma}[{\cite[Proposition 2.1]{HLL23}}]\label{pro:d_r}
For each $ r $ and $ 1\leq r \leq s$, if the dimension of $ C_{D} $
is $ s$, then
\begin{equation}\label{eq:d_r}
d_{r}(C_{D})= n-\max\big\{|D \cap H|: H \in [\mathbb{F},s-r]_{p}\big\}.
\end{equation}
\end{lemma}

It should remarked that, for a $p$-ary quadratic function $f(x)$ defined over $\F_{p^s}$, the weight hierarchy of $C_{D_f^a}$ in \eqref{defcode0} was determined by Wan \cite{27WZ94,WW97} for $a=0$ and Li et al. \cite{19LF17,LL22} for $a\in\F_p^*$, where
$D_f^a=\{x\in\F_{p^s}:f(x)=a\}$. The weight hierarchy of $C_{D_\alpha}$ in \eqref{defcode2}, where $D_\alpha = \{(x,y)\in \F_{p^s}^2\backslash\{(0,0)\}:f(x)+\Tr^s(\alpha y)=0\} $ and $f(x)=\sum\limits_{i=0}^{s-1}\Tr^s(a_ix^{p^i+1})$ non-degenerate, and the weight hierarchy of $C_{D_\alpha}$ in \eqref{defcode1}, where $D_\alpha = \{(x,y)\in \F_{p^s}^2\backslash\{(0,0)\}:f(x)+\Tr^s(\alpha y)=0\}$ were both determined completely by Li et al. \cite{LL22,HLL23}.

In Section 2, we set the main notation and give some
properties of $p$-ary quadratic forms. In Section 3, we introduce the parameters of 3-weight linear codes derived from two $p$-ary quadratic forms over different finite fields and determine their weight hierarchies completely. Section 4 completes the paper.

\section{Preliminaries}

\subsection{Some notations fixed throughout this paper}

  For convenience, we fix the following notations. One is referred to \cite{IR90} for basic results on cyclotomic field $ \mathbb{Q}(\zeta_{p}).$
\begin{itemize}
\item Let $\Tr^s$ be the trace function from $\mathbb{F}_{q}$ to $\mathbb{F}_{p}$.
Namely, for each $x\in \mathbb{F}_{q}$,
$$
\Tr^s(x)=x+x^{p}+ \cdots +x^{p^{s-1}}.
$$
\item $p^{\ast}=(-1)^{\frac{p-1}{2}}p,\ \zeta_{p}=\exp(\frac{2\pi i}{p})$.
\item $\upsilon$ is a function on $\F_p$ satisfying $\upsilon(0)=p-1$ and $\upsilon(z)=-1$ for $z\in \mathbb{F}_{p}^{\ast}$.
\item $\bar{\eta}$ is the quadratic character of $\mathbb{F}_{p}^{\ast}$.
It is extended by letting $\bar{\eta}(0)=0$.
\item Let $\mathbb{Z}$ be the rational integer ring and $\mathbb{Q}$ be the rational field. Let $\mathbb{K}$ be the cyclotomic
field $\mathbb{Q}(\zeta_{p})$. The field extension $\mathbb{K}/\mathbb{Q}$ is Galois of degree $p-1$.
The Galois group $\mathrm{Gal}(\mathbb{K}/\mathbb{Q})=\Big\{\sigma_{z}: z\in (\mathbb{Z}/ p\mathbb{Z})^{\ast}\Big\}$, where
$\sigma_{z}$ is defined by $\sigma_{z}(\zeta_{p})=\zeta_{p}^{z}$.
\item $\sigma_{z}(\sqrt{p^{\ast}})=\bar{\eta}(z)\sqrt{p^{\ast}}$, for $1\leq z \leq p-1$.
\item Let $\Big\langle\alpha_{1},\alpha_{2},\cdots,\alpha_{r}\Big\rangle$ denote a space spanned by $\alpha_{1},\alpha_{2},\cdots,\alpha_{r}$.
\end{itemize}

\begin{lemma}[{\cite[Lemma 4]{19TXF17}}]\label{lem:7}
 With the symbols and notations above, for any $z\in \mathbb{F}_{p}$, we have the following.

$$
\sum\limits_{y\in \mathbb{F}_{p}^{\ast}}\sigma_{y}((p^{\ast})^{\frac{r}{2}}\zeta_{p}^{z})
=\left\{\begin{array}{ll}
\bar{\eta}(-z)p^r(p^{\ast})^{-\frac{r-1}{2}}, & \textrm{if\ } r  \ \textrm{is odd\ }, \\
\upsilon(z)p^r(p^{\ast})^{-\frac{r}{2}},  & \textrm{if\ } r \ \textrm{is even\ }.
\end{array}
\right.
$$

\end{lemma}

\subsection{Quadratic form}

Viewing $\mathbb{F}_{q}$ with $q=p^s$ as an $\mathbb{F}_{p}$-linear space and fixing $\upsilon_{1},\upsilon_{2},\cdots,\upsilon_{s} \in \mathbb{F}_{q} $ as its $\F_p$-basis.
There is an $\F_p$-linear isomorphism $\F_q\simeq\F_p^s$ defined as
$$x=x_{1}\upsilon_{1}+x_{2}\upsilon_{2}+\cdots+x_{s}\upsilon_{s} \mapsto X=(x_1,x_2,\cdots,x_s),$$
where $X\in \F_p^s$ is called the coordinate vector of $x$ under the basis $v_1,v_2,\cdots,v_s$ of $\F_q$.

Let $f$ be a quadratic form over $\mathbb{F}_{q}$ with values in $\mathbb{F}_{p}$ and \begin{equation*}\label{eq:F}
    F(x,y)=\frac{1}{2}\Big(f(x+y)-f(x)-f(y)\Big), \textrm{for any}\  x,y\in\mathbb{F}_q,
\end{equation*} then $f$ can be represented by
\begin{align}\label{eq:f}
f(x)=f(X)=f(x_{1},x_{2},\cdots,x_{s})&=\sum_{1\leq i,j\leq s}F(v_i,v_j)x_{i}x_{j}
=XAX^T,
\end{align}
where $A=(F(v_i,v_j))_{s\times s}, F(v_i,v_j)\in \mathbb{F}_{p}, F(v_i,v_j)=F(v_j,v_i)$ and $X^T$ is the transposition of $X$. Denote by $R_f=\textrm{Rank}A$ the rank of $f$, we say that $f$ is non-degenerate if $R_f=s$ and degenerate, otherwise. We can find an invertible matrix $M$ over $\F_p$ such that
$$
MAM^T=\textrm{diag}(\lambda_{1},\lambda_{2},\cdots,\lambda_{R_{f}},0,\cdots,0)
$$
is a diagonal matrix, where $\lambda_{1},\lambda_{2},\cdots,\lambda_{R_{f}}\in\F_p^*$.
Let $\Delta_{f}=\lambda_{1}\lambda_{2}\cdots\lambda_{R_{f}}$, and $\Delta_{f}=1$ if $R_{f}=0.$
We call $\bar{\eta}(\Delta_{f})$ the sign $\varepsilon_{f}$ of $f$,
which is an invariant under nonsingular linear transformations in matrix.

For a subspace $H$ of $\F_q$, its dual space $H^{\perp_f}$ is defined by
$$
H^{\perp_f}=\Big\{x\in \mathbb{F}_{q}:\ F(x,y)=0, \ \mbox{for  any} \  y \in H\Big\}.
$$
Restricting the quadratic form $f$ to $H$,
it becomes a quadratic form denoted by $f|_{H}$ over $H$ in $r$ variables.
Let $R_{H}$ and $\varepsilon_{H}$ be the rank and sign of $f|_{H}$ over $H$, respectively.

For $\beta\in \mathbb{F}_{p}$, set $D_{\beta}=\Big\{x\in \mathbb{F}_{q}|f(x)=\beta\Big\}$.
There are some lemmas essential to prove our main results.

\begin{lemma}[{\cite[Lemma 2]{LL20-0}}]\label{lem:1}
Let $f$ be a quadratic form over $\mathbb{F}_{q}, \beta\in \mathbb{F}_{p}$ and $H$ be an $r$-dimensional nonzero subspace of $\mathbb{F}_{q}$, then
$$
|H\cap D_{\beta}|=\left\{\begin{array}{ll}
p^{r-1}\Big(1+\upsilon(\beta)\varepsilon_{H}(p^*)^{-\frac{R_{H}}{2}}\Big),  &\textrm{if\ } \ R_{H}\equiv0\pmod 2, \\
p^{r-1}\Big(1+\overline{\eta}(\beta)\varepsilon_{H}(p^*)^{-\frac{R_{H}-1}{2}}\Big),  &\textrm{if\ } \ R_{H}\equiv1\pmod 2,
\end{array}
\right.
$$
where $\upsilon(\beta)=p-1$ if $\beta=0$, otherwise $\upsilon(\beta)=-1$.
\end{lemma}

\begin{lemma}[{\cite[Lemma 2.5]{HLL23}}]\label{lem:2}
Let $f$ be a quadratic form over $\mathbb{F}_{q}$ with the rank $R_f$. There exists an $e_f$-dimensional subspace
$H $ of $ \mathbb{F}_{q}$ such that $\F_q^{\perp_f}\subseteq H$ and $f(x)=0$ for any $x\in H$, where

$
e_f=\left\{\begin{array}{ll}
s-\frac{R_f+1}{2}, & \textrm{if $R_f$ is odd}, \\
s-\frac{R_f}{2}, & \textrm{if $R_f$ is even and $\varepsilon_{f}=\overline{\eta}(-1)^{\frac{R_f}{2}}$},\\
s-\frac{R_f+2}{2}, & \textrm{if $R_f$ is even and $\varepsilon_{f}=-\overline{\eta}^{\frac{R_f}{2}}$}.
\end{array}
\right.
$
\end{lemma}

\begin{lemma}[{\cite[Lemma 2.6]{HLL23}}]\label{lem:2-2}
Let $f$ be a quadratic form over $\mathbb{F}_{q}$ with the rank $R_f$. For each $a\in \F_p^*$, there exists an $l_0$-dimensional subspace
$H$ of $ \mathbb{F}_{q}$ such that $R_H=1, \varepsilon_{H}=\bar{\eta}(a)$ and $H \cap\ \F_q^{\perp_f}=\{0\}$, where $$
l_0=\left\{\begin{array}{ll}
\frac{R_f-1}{2}, & \textrm{if $R_f$ is odd}, \\
\frac{R_f}{2}, & \textrm{if $R_f$ is even }.
\end{array}
\right.
$$
\end{lemma}
\begin{remark} In fact, when $R_f$ is odd, by Proposition 2 \cite{19LF17}, we can construct an $\frac{R_f+1}{2}$-dimensional subspace $H$ of $\F_q$ such that $R_H=1, \varepsilon_{H}=\overline{\eta}(-1)^{\frac{R_f-1}{2}}\varepsilon_{f}$ and $H \cap\ \F_q^{\perp_f}=\{0\}$, which concludes that $|H\cap D_{\beta}|=2p^{\frac{R_f-1}{2}}$ for any $a\in \F_p^*$ and
$\overline{\eta}(a)=\overline{\eta}(-1)^{\frac{R_f-1}{2}}\varepsilon_{f}$ by Lemma \ref{lem:1}.

\end{remark}

Let $f$ be a quadratic form over $\mathbb{F}_{q}$ defined by eq.\eqref{eq:f}. For any $x,y \in\F_q$, there exists a linearized polynomial $L_f$ over $\mathbb{F}_q$ such that $f(x)=\Tr^s(xL_f(x))$ and
\begin{equation}\label{eq:F} F(x,y)=\Tr^s\Big(xL_{f}(y)\Big)=\Tr^s\Big(yL_{f}(x)\Big). \end{equation}
Let $S_f=\mathrm{Im}(L_{f})=\Big\{L_{f}(x):x\in\F_q\Big\},\ \Ker(L_{f})=\Big\{x\in\F_q:L_{f}(x)=0\Big\}$ denote the image and the kernel of $L_{f}$, respectively.
If $b\in S_f$, we denote $x_{b}\in \mathbb{F}_q$ with $L_{f}(x_{b})=-\frac{b}{2}$. For more details, one can refer to \cite[Quadratic form]{HLL23}.

From eq. \eqref{eq:F}, we have
\begin{equation*}
\Ker(L_{f})=\{x\in \mathbb{F}_q:f(x+y)=f(x)+f(y), \textrm{for any}\  y\in\mathbb{F}_q\}=\mathbb{F}_q^{\perp_f}
\end{equation*}
and $\textrm {rank}\ L_f = R_f$.

The following lemmas will play an important role in settling the weight hierarchies.

\begin{lemma}[{\cite[Lemma 5]{19TXF17}}]\label{lem:6}
 Let the symbols and notation be as above and $f$ be defined in \eqref{eq:f} and $b\in \mathbb{F}_{q}$. Then
$$
\sum\limits_{x\in \mathbb{F}_{q}}\zeta_{p}^{f(x)-\Tr^s(bx)}
=\left\{\begin{array}{ll}
0, & \textrm{if\ } b\notin S_f, \\
\varepsilon_{f}(p^{\ast})^{\frac{R_{f}}{2}}p^{s-R_f}\zeta_{p}^{-f(x_{b})},  & \textrm{if\ } b\in S_f.
\end{array}
\right.
$$
where $x_{b}$ satisfies $L_{f}(x_{b})=-\frac{b}{2}$.
\end{lemma}


\begin{lemma}\label{lem:0}
Let $f$ be a quadratic form over $\mathbb{F}_{q}$ and $H$ be an $r$-dimensional subspace of $ \mathbb{F}_{q}$ with the rank $R_H$ and sign $\varepsilon_{H}$. We have
$$\sum\limits_{x\in H}\zeta_{p}^{f(x)}
=\varepsilon_{H}(p^{\ast})^{\frac{R_{H}}{2}}p^{r-R_H}.
$$
\end{lemma}

\begin{proof} We only prove the case of $R_{H}$ odd. Choosing a non-square element $z\in\F_{p}^*$.
By Lemma \ref{lem:1}, we have \begin{align*}
    \sum\limits_{x\in H}\zeta_{p}^{f(x)}&= \sum\limits_{\beta\in \F_p}\sum\limits_{x\in H\cap D_{\beta}}\zeta_{p}^{f(x)}\\
    &=\sum\limits_{\beta\in \F_p}|H\cap D_{\beta}|\zeta_{p}^{\beta}\\
&=|H\cap D_{0}|+\frac{1}{2}|H\cap D_{1}|\sum\limits_{x\in \F_p^*}\zeta_{p}^{x^2}+\frac{1}{2}|H\cap D_{z}|\sigma_{z}(\sum\limits_{x\in \F_p^*}\zeta_{p}^{x^2})\\
&=|H\cap D_{0}|+\frac{1}{2}|H\cap D_{1}|(\sqrt{p^{\ast}}-1)+\frac{1}{2}|H\cap D_{z}|(-\sqrt{p^{\ast}}-1)\\
&=\varepsilon_{H}(p^{\ast})^{\frac{R_{H}}{2}}p^{r-R_H},
\end{align*}
where the last second equality comes from $\sum\limits_{x\in \F_p}\zeta_{p}^{x^2}=\sqrt{p^{\ast}}$ by Lemma \ref{lem:6}.
\end{proof}











\section{The weight hierarchy of the presented linear code}

In this subsection, we give the weight hierarchy of $C_{D}$ in \eqref{defcode1}.

The following Lemma \ref{thm:wd} gives the the weight distribution of the code $C_{D}$ be defined in \eqref{defcode1}.

\begin{lemma}\label{thm:wd} Let $D$ be defined in \eqref{set:D1} and the code $C_{D}$ be defined in \eqref{defcode1}. Denote $s=s_1+s_2, R=R_f+R_g,\varepsilon=\varepsilon_{f}\varepsilon_{g}$.
Then the code $C_{D}$ is an $[n,s]$ linear code over $ \mathbb{F}_{p} $
with the weight distribution in Tables 1 and 2, where
$$n=\left\{\begin{array}{ll}
p^{s-1}-1, & \textrm{if\ } 2\nmid R  , \\
p^{s-1}-1+(p-1)p^{s-1}\varepsilon_{f}\varepsilon_{g}(p^{\ast})^{-\frac{R}{2}},  & \textrm{if\ } 2\mid R  .
\end{array}
\right.$$
\begin{table}\label{tab:wd:o}
\centering
\caption{The weight distribution of $C_{\D}$ of Lemma \ref{thm:wd} when $R$ is odd}
\begin{tabular*}{10.5cm}{@{\extracolsep{\fill}}ll}
\hline
\textrm{Weight} $\omega$ & \textrm{Multiplicity} $A_\omega$   \\
\hline
0 &   1  \\
$(p-1)p^{s-2}$ &  $p^{s}-p^{R}+p^{R-1}-1$  \\
$(p-1)p^{s-2}(1-\varepsilon(p^*)^{-\frac{R-1}{2}})$  & $\frac{1}{2}(p-1)p^{R-1}(1+\varepsilon(p^*)^{-\frac{R-1}{2}})$  \\
$(p-1)p^{s-2}(1+\varepsilon(p^*)^{-\frac{R-1}{2}})$  & $\frac{1}{2}(p-1)p^{R-1}(1-\varepsilon(p^*)^{-\frac{R-1}{2}})$  \\
\hline
\end{tabular*}
\end{table}
\begin{table}
\centering
\caption{The weight distribution of $C_{D}$ of Lemma \ref{thm:wd} when $R$ is even}
\begin{tabular*}{10.5cm}{@{\extracolsep{\fill}}ll}
\hline
\textrm{Weight} $\omega$ \qquad& \textrm{Multiplicity} $A_\omega$   \\
\hline
0 \qquad&   1  \\
$(p-1)p^{s-2}$ \qquad&  $p^{R-1}-1+\varepsilon(1-p^{-1})(p^*)^{\frac{R}{2}}$  \\
$(p-1)p^{s-2}(1+\varepsilon(p-1)(p^*)^{-\frac{R}{2}})$  \qquad& $ p^{s}-p^{R}$  \\
$(p-1)p^{s-2}(1+\varepsilon p(p^*)^{-\frac{R}{2}})$  \qquad& $(1-p^{-1})p^{R}(1-\varepsilon(p^*)^{-\frac{R}{2}})$  \\
\hline
\end{tabular*}
\end{table}
\end{lemma}

\begin{remark} By Lemma \ref{lem:6}, it's easy to check that each $p$-ary quadratic function $f$ over $\F_{p^s}$ is weakly regular unbalance $(s-R_f)$-plateaued function with index $2$ and A. S{\i}nak \cite{S22} studied linear codes derided from two weakly regular unbalance plateaued function and determined their weight distributions. So, we omit the proof of Lemma \ref{thm:wd}.
\end{remark}

We restrict $p$-ary quadratic function $f$ over $\F_q$ to subspaces and quotient space $\F_{q}\Big/\F_{q}^{\perp_f}$ of $\F_q$, for more details one can refer to \cite{19LF17,LL22}, and obtain the following interesting results. 

\begin{lemma}\label{lem:main}
Let $f(x),g(y)$ be quadratic forms over $\mathbb{F}_{q_i},i=1,2$ defined in \eqref{eq:f} respectively. For $(u,v)\in \mathbb{F}$ and $t\in \F_p$ and $H_r=L_f(\overline{H}_{r_1})\times L_g(\overline{H}_{r_2})$, where $\overline{H}_{r_1}$ (resp $\overline{H}_{r_2}$) is an $r_1$ (resp $r_2$) - dimensional subspace of $\F_{q_1}\Big/\F_{q_1}^{\perp_f}$ (resp $\F_{q_2}\Big/\F_{q_2}^{\perp_g}$) and $r=r_1+r_2$, define $S_{fg,H_r}(t)=\#\{(u,v)\in H_r:f(x_u)+g(y_v)=t\}$. Write $ \varepsilon_{f,1}=\varepsilon_{f,H_{r_1}}, \varepsilon_{g,2}=\varepsilon_{g,H_{r_2}}, R_{f,1}=R_{f,H_{r_1}}, R_{g,2}=R_{g,H_{r_2}}, R_{H_r}=R_{f,1}+R_{g,2}$. Then, we have
the following.

  (1) When $R_{H_r}$ is odd,
  \begin{equation}
  S_{fg,H_r}(t)=p^{r-1}\Big(1 + \varepsilon_{f,1}\varepsilon_{g,2}(p^{\ast})^{-\frac{R_{H_r}-1}{2}}\bar{\eta}(t)\Big).
  \end{equation}

  (2) When $R_{H_r}$ is even,
  \begin{equation}
  S_{fg,H_r}(t)=p^{r-1}\Big(1 + \varepsilon_{f,1}\varepsilon_{g,2}(p^{\ast})^{-\frac{R_{H_r}}{2}}\upsilon(t)\Big).
  \end{equation}

Here $x_{u}, y_v$ satisfy $L_{f}(x_{u})=-\frac{u}{2}, L_g(y_v)=-\frac{v}{2}$.
\end{lemma}

\begin{proof} From the orthogonality of exponential sums and Lemma \ref{lem:0}, we have
\begin{align*}
  &S_{fg,H_r}(t)=\frac{1}{p} \sum\limits_{u\in L_f(H_{r_1})} \sum\limits_{v\in L_g(H_{r_2})} \sum\limits_{z\in \F_p} \zeta^{z(f(x_u)+g(y_v)-t)}\\
&=p^{r-1} + \frac{1}{p} \sum\limits_{z\in \F_p^*}\zeta^{-zt}\sigma_{z}\Big(\sum\limits_{x\in \overline{H}_{r_1}}\zeta^{f(x)} \sum\limits_{y\in \overline{H}_{r_2}}\zeta^{g(y_v)}\Big)\\
&=p^{r-1} + \frac{1}{p} \sum\limits_{z\in \F_p^*}\zeta^{-zt}\sigma_{z}\Big(\varepsilon_{f,1}\varepsilon_{g,2}(p^*)^{\frac{R_{f,1}+R_{g,2}}{2}}p^{r-(R_{f,1}+R_{g,2})}\Big)\\
&=p^{r-1} + \varepsilon_{f,1}\varepsilon_{g,2}p^{r-(R_{f,1}+R_{g,2})-1} \sum\limits_{z\in \F_p^*}\sigma_{z}\Big(\zeta^{-t}(p^*)^{\frac{R_{f,1}+R_{g,2}}{2}}\Big)\\
&=\left\{\begin{array}{ll}
p^{r-1} + \varepsilon_{f,1}\varepsilon_{g,2}p^{r-1}\bar{\eta}(t)(p^{\ast})^{-\frac{R_{f,1}+R_{g,2}-1}{2}}, & \textrm{if\ } R_{H_r}  \ \textrm{is odd\ }, \\
p^{r-1} + \varepsilon_{f,1}\varepsilon_{g,2}p^{r-1}\upsilon(t)(p^{\ast})^{-\frac{R_{f,1}+R_{g,2}}{2}},  & \textrm{if\ } R_{H_r} \ \textrm{is even\ }.
\end{array}
\right.
\end{align*}
\end{proof}


By Lemma~\ref{thm:wd}, we know that the dimension of the code $C_{D}$ defined in \eqref{defcode1} is $s$.
So, by Lemma \ref{pro:d_r}, we give a general formula,
that is
\begin{align}
        d_{r}(C_{D})&=n-\max\Big\{|H_r^\perp\cap D|: H_r \in [\mathbb{F},r]_{p}\Big\}  \\
        &=n-\max\Big\{|H_{s-r}\cap D|: H_{s-r} \in [\mathbb{F},s-r]_{p}\Big\}\label{eq:d_r:2},
\end{align}
which will be employed to calculate the generalized Hamming weight $d_r(C_D)$, and here

$H_r^\perp = \Big\{(x,y)\in\F:\Tr^{s_1}(ux)+\Tr^{s_2}(vy) =0, \textrm{for any $(u,v)\in H_r$}\Big \}$.

\begin{lemma} \label{lem:d_r:2}
Let $f(x),g(y)$ be quadratic forms over $\mathbb{F}_{q_i},i=1,2$ defined in \eqref{eq:f} with the sign $\varepsilon_f, \varepsilon_g$ and the rank $R_f, R_g$ respectively. Let $H_r$ be an $r$-dimensional subspace of $\mathbb{F}$, define $N(H_r)=\Big\{(x,y)\in \mathbb{F}: f(x)+g(y)=0, \Tr^{s_1}(ux)+\Tr^{s_2}(vy) =0, \ \textrm{ for any}\ (u,v)\in H_r\Big\}$. Define $R= R_f+R_g$. We have the following.

(1) If $R$ is even,
$$|N(H_r)|=
p^{s-r-1}\Big(1+\varepsilon_f \varepsilon_g(p^*)^{-\frac{R}{2}}\sum\limits_{(u,v)\in H_r\bigcap (S_f\times S_g)}\upsilon(f(x_u)+g(y_v))\Big),$$

(2)If $R$ is odd,
$$
|N(H_r)|=
p^{s-r-1}\Big(1+\varepsilon_f \varepsilon_g(p^*)^{-\frac{R-1}{2}}\sum\limits_{(u,v)\in H_r\bigcap (S_f\times S_g)}\bar{\eta}(f(x_u)+g(y_v))\Big),
$$
where $x_{u},y_v$ satisfy $L_{f}(x_{u})=-\frac{u}{2}, L_{g}(y_{v})=-\frac{v}{2}$.
\end{lemma}

\begin{proof} Denote by $H_{r}^\star=H_{r}\setminus\{(0,0)\}$.
By the orthogonal property of additive characters, we have
\begin{align*}
&p^{r+1}|N(H_r)|
=\sum_{(x,y)\in \F}\sum_{z\in \F_p}\zeta_{p}^{z f(x)+zg(y)}\sum_{(u,v)\in H_r}\zeta_{p}^{\Tr^{s_1} (ux)+\Tr^{s_2} (vy)} \\
&=\sum_{(x,y)\in \F}\sum_{(u,v)\in H_r}\zeta_{p}^{\Tr^{s_1} (ux)+\Tr^{s_2} (vy)}
+\sum_{(x,y)\in \F}\sum_{z\in \F_p^{\ast}}\sum_{(u,v)\in H_r}\zeta_{p}^{z f(x)+\Tr^{s_1} (ux)+\Tr^{s_2} (vy)} \\
&=p^{s}+\varepsilon_f \varepsilon_g p^{s-R}\sum\limits_{(u,v)\in H_r\bigcap (S_f\times S_g)}\sum\limits_{z\in\F_p^*} \sigma_z((p^*)^{\frac{R}{2}}\zeta^{-(f(x_u)+g(y_v))}),
\end{align*}
where the last equation comes from Lemma \ref{lem:6} and
\begin{align*}
\sum_{(x,y)\in \F}\sum_{(u,v)\in H_{r}}\zeta_{p}^{\Tr^{s_1} (ux)+\Tr^{s_2} (vy)}
&=\sum_{(x,y)\in \F} \ 1+\sum_{(x,y)\in \F}\sum_{(u,v)\in H_{r}^\star}\zeta_{p}^{\Tr^{s_1} (ux)+\Tr^{s_2} (vy)}  \\
&=p^{s}+\sum_{(u,v)\in H_{r}^\star}\sum_{(x,y)\in \F}\zeta_{p}^{\Tr^{s_1}(ux)+\Tr^{s_2}(vy)}=p^{s}.
\end{align*}

So, the desired result is directly from Lemma \ref{lem:7}. Thus, we complete the proof.
\end{proof}

Let $H_r$ be an $r$-dimensional subspace of $\mathbb{F}$, it's easy to see that $|N(H_r)|=|H_r^\perp\cap\D |+1$. Hence, we have
\begin{equation}\label{eq:d_r:3}
     d_{r}(C_{\D})=n+1-\max\Big\{|N(H_r)|: H_r \in [\mathbb{F},r]_{p}\Big\}.
\end{equation}

In the following, we shall determine the weight hierarchy of $C_{D}$ in \eqref{defcode1}  by calculating $N(H_r)$ in Lemma~\ref{lem:d_r:2} and $|H_{s-r}\cap D|$ in \eqref{eq:d_r:2}.

\begin{theorem}\label{thm:wh:e1} Let $f(x),g(y)$ be quadratic forms over $\mathbb{F}_{q_i},i=1,2$ defined in \eqref{eq:f} with the sign $\varepsilon_f, \varepsilon_g$ and the rank $R_f, R_g$ respectively. Let $D$ be defined in \eqref{set:D1} and the code $C_{D}$ be defined in \eqref{defcode1}. Suppose that $R= R_f+R_g$ is even and $\varepsilon_{f}\varepsilon_{g}=\bar{\eta}(-1)^{\frac{R}{2}}$. Then we have the following.

$$ d_{r}(C_{D})=\left\{\begin{array}{ll}
p^{s-1}-p^{s-1-r},   0< r\leq \frac{R}{2}-1, \\
p^{s-1}-p^{s-r}+(p-1)p^{s-1-\frac{R}{2}}, \frac{R}{2} \leq r \leq s.
\end{array}
\right.
$$
\end{theorem}

\begin{proof} We only prove the case of $R_f$ even and $\varepsilon_f=-\overline{\eta}(-1)^{\frac{R_f}{2}}$. For $f$ and $g$, let $e_f$ and $e_g$ be defined in Lemma \ref{lem:2}. In this case, $e_f=s_1-\frac{R_f+2}{2}$ and $e_g=s_2-\frac{R_g+2}{2}$.

When $s-(e_f+e_g+2) \leq r \leq s$, then $0 \leq s-r \leq e_f+e_g+2$.
By Lemma~\ref{lem:2}, there exists an $e_f$-dimensional subspace $J_{e_f}$ of $\F_{q_1}$ such that $f(x)=0$ for any $x\in J_{e_f}$, and an $e_g$-dimensional subspace $J_{e_g}$ of $\F_{q_2}$ such that $g(y)=0$ for any $y\in J_{e_g}$. By Lemma \ref{lem:2-2} and its remark, we can take $(u',v'), (u'',v'')\in D$, such that $J_{e_f}^{\perp_f}=J_{e_f}\bigoplus <u',u''>, J_{e_g}^{\perp_g}= J_{e_g}\bigoplus <v',v''> $, where $ f(u'+u'')=f(u')+f(u'')$ and $ g(v'+v'')=g(v')+g(v'')$.
Note that the dimension of the subspace $J_{e_f}\times J_{e_g}$ is $e_f+e_g$. Let $H_{s-r}$ be an $(s-r)$-dimensional subspace of $(J_{e_f}\times J_{e_g})\bigoplus \Big<(u',v'),(u'',v'')\Big>$,
then, $|H_{s-r}\cap D|=p^{s-r}-1$.
Hence, by \eqref{eq:d_r:2}, we have
$$
d_{r}(C_{\D})=n-\max\Big\{|D \cap H|: H \in [\mathbb{F},s-r]_{p}\Big\}=p^{s-1}-p^{s-r}+(p-1)p^{s-1-\frac{R}{2}}.
$$
Thus, it remains to determine $d_r(C_{\D})$ when $0< r\leq s-(e_f+e_g)-3$.

 When $0< r\leq s-(e_f+e_g)-3$, that is, $0< r\leq \frac{R}{2}-1$. Suppose $H_r$ is an $r$-dimensional subspace of $\F$. Recall that $v(0)=p-1$ and $v(x)=-1$ for $x\in\mathbb{F}_{p}^{\ast}$.
By Lemma~\ref{lem:d_r:2}, we have
\begin{align*}
N(H_r)&=p^{s-r-1}\Big(1+\varepsilon_f \varepsilon_g(p^*)^{-\frac{R}{2}}\sum\limits_{(u,v)\in H_r\bigcap (S_f\times S_g)}\upsilon(f(x_u)+g(y_v))\Big)\\
&=p^{s-r-1}\Big(1+p^{-\frac{R}{2}}\sum\limits_{(u,v)\in H_r\bigcap (S_f\times S_g)}\upsilon(f(x_u)+g(y_v))\Big).
\end{align*}

For $0< r\leq \frac{R}{2}-2$, let $J_{f,r_1}$ be an $r_1$-dimensional subspace of $J_{e_f}\Big/\F_{q_1}^{\perp_f}$, and $J_{g,r_2}$ be an $r_2$-dimensional subspace of $J_{e_g}\Big/\F_{q_2}^{\perp_g}$, where $1\leq r_1\leq \frac{R_f-2}{2}, 0\leq r_2\leq \frac{R_g-2}{2}$ and $r_1+r_2=r$ .
Taking \begin{equation}\label{eq:H} H_{r}=L_f(J_{f,r_1})\times L_g(J_{g,r_2}),\end{equation} then $N(H_r)$ reaches its maximum
$$ N(H_r)=p^{s-(r+1)}\Big(1+(p-1)p^{r-\frac{R}{2}}\Big) = p^{s-1-r} + (p-1)p^{s-1-\frac{R}{2}}. $$
So, for $0< r\leq\frac{R}{2}-2$, the desired result is obtained by Lemma~\ref{lem:d_r:2} and \eqref{eq:d_r:3}.

For $r=\frac{R}{2}-1$, define $H_r=H_{\frac{R}{2}-2}\oplus \Big\langle(L_f(u'),L_g(v')) \Big\rangle$, where $H_{\frac{R}{2}-2}$ is defined as \eqref{eq:H} and $(u',v')$ is defined as above, then $N(H_r)$ reaches its maximum
$$ N(H_r)=p^{s-(r+1)}\Big(1+(p-1)p^{r-\frac{R}{2}}\Big) = p^{s-1-r} + (p-1)p^{s-1-\frac{R}{2}}. $$
So, for $r=\frac{R}{2}-1$, the desired result is obtained by Lemma~\ref{lem:d_r:2} and \eqref{eq:d_r:3}.

\end{proof}

\begin{remark} Recall the Griesmer-like bound on GHWs of linear codes \cite{TV95}. Let $C$ be an $[n, k]$ linear code over $\F_p$. For any integer $1 \leq r \leq k$, it is
known that
$$d_r(C) \geq \sum\limits_{i=0}^{r-1}\lceil\frac{d_1(C)}{p^i}\rceil.$$
 It is easy to check that the codes $C_{D}$  in Theorem \ref{thm:wh:e1} satisfy
that $d_r(C) = \sum\limits_{i=0}^{r-1}\lceil\frac{d_1(C)}{p^i}\rceil$, for $1 \leq r \leq \frac{R}{2}$.
\end{remark}

\begin{theorem}\label{thm:wh:e2} Let $f(x),g(y)$ be quadratic forms over $\mathbb{F}_{q_i},i=1,2$ defined in \eqref{eq:f} with the sign $\varepsilon_f, \varepsilon_g$ and the rank $R_f, R_g$ respectively. Let $D$ be defined in \eqref{set:D1} and the code $C_{D}$ be defined in \eqref{defcode1}. Suppose that $R= R_f+R_g$ is even and $\varepsilon_{f}\varepsilon_{g}=-\bar{\eta}(-1)^{\frac{R}{2}}$.
 Then we have the following.

  $$
d_{r}(C_{D})=\left\{\begin{array}{ll}
p^{s-2}(p-1)(1-p^{1-\frac{R}{2}}), r=1\\
p^{s-1}-p^{s-1-r}-(p^2-1)p^{s-2-\frac{R}{2}}, 1< r\leq\frac{R}{2},\\
p^{s-1}-p^{s-r}-(p-1)p^{s-1-\frac{R}{2}}, \frac{R}{2}+1 \leq r \leq s.
\end{array}
\right.
$$

\end{theorem}

\begin{proof}  For $f$ and $g$, let $e_f$ and $e_g$ be defined in Lemma \ref{lem:2}.
The proof of Case $s-(e_f+e_g) \leq r \leq s$ is similar to that of Theorem \ref{thm:wh:e1}, we omit it here.

We only prove the case of $R_f$ is even and $0< r\leq s-(e_f+e_g)-1$. Without loss of generality, we assume that $\varepsilon_f=-\bar{\eta}(-1)^{\frac{R_f}{2}}$. In this case, $0< r\leq \frac{R}{2}$.

Suppose $H_r$ is an $r$-dimensional subspace of $\F$.
By Lemma~\ref{lem:d_r:2}, we have
\begin{align*}
N(H_r)&=p^{s-r-1}\Big(1+\varepsilon_f \varepsilon_g(p^*)^{-\frac{R}{2}}\sum\limits_{(u,v)\in H_r\bigcap (S_f\times S_g)}\upsilon(f(x_u)+g(y_v))\Big)\\
&=p^{s-r-1}\Big(1-p^{-\frac{R}{2}}\sum\limits_{(u,v)\in H_r\bigcap (S_f\times S_g)}\upsilon(f(x_u)+g(y_v))\Big).
\end{align*}

For $r=1$, it's easy to see that the maximum of $N(H_r)$ is equal to $p^{s-2}$, which concludes that $d_{1}(C_{D})=p^{s-2}(p-1)(1-p^{1-\frac{R}{2}})$.

For $2\leq r\leq\frac{R}{2}$, by Lemma \ref{lem:2-2}, there exists an $r_1$-dimensional subspace
$J_{f,r_1}$ of $ \mathbb{F}_{q_1}$ such that $R_{J_{f,r_1}}=1, \varepsilon_{J_{f,r_1}}=\bar{\eta}(\beta)$ and $J_{f,r_1}\cap \F_{q_1}^{\perp_f}=\{0\}$ for some non-square $\beta\in\F_p^*$, and $r_2$-dimensional subspace
$J_{g,r_2}$ of $ \mathbb{F}_{q_2}$ such that $R_{J_{g,r_2}}=1, \varepsilon_{J_{g,r_2}}=\bar{\eta}(-1)$ and $J_{g,r_2}\cap \F_{q_2}^{\perp_g}=\{0\}$,
 where $1\leq r_1\leq \frac{R_f}{2}, 1\leq r_2\leq \frac{R_g}{2}$ and $r_1+r_2=r$ .
Taking \begin{equation}\label{eq:H2} H_{r}=L_f(J_{f,r_1})\times L_g(J_{g,r_2}),\end{equation} then by Lemma \ref{lem:main} $N(H_r)$ reaches its maximum
$$ N(H_r)=p^{s-(r+1)}\Big(1+(p-1)p^{r-1-\frac{R}{2}}\Big) = p^{s-1-r} + (p-1)p^{s-2-\frac{R}{2}}. $$
So, for $2\leq r\leq\frac{R}{2}$, the desired result is obtained by Lemma~\ref{lem:d_r:2} and \eqref{eq:d_r:3}.
\end{proof}

\begin{theorem}\label{thm:wh:o} Let $f(x),g(y)$ be quadratic forms over $\mathbb{F}_{q_i},i=1,2$ defined in \eqref{eq:f} with the sign $\varepsilon_f, \varepsilon_g$ and the rank $R_f, R_g$ respectively. Let $D$ be defined in \eqref{set:D1} and the code $C_{D}$ be defined in \eqref{defcode1}. Suppose that $R= R_f+R_g$ is odd.
 Then we have the following.
  $$
d_{r}(C_{D})=\left\{\begin{array}{ll}
p^{s-1}-p^{s-1-r} - (p-1)p^{s-2-\frac{R-1}{2}},   0< r\leq \frac{R-1}{2}, \\
p^{s-1}-p^{s-r}, \frac{R+1}{2} \leq r \leq s.
\end{array}
\right.
$$

\end{theorem}

\begin{proof}
The proof of Case $s-(e_f+e_g) \leq r \leq s$ is similar to that of Theorem \ref{thm:wh:e1}, we omit it here.

 When $0< r\leq s-(e_f+e_g)-1$, we only prove the case of $\varepsilon_{f}\varepsilon_{g}=\bar{\eta}(-1)^{\frac{R-1}{2}}$ and $R_f$ even and $\varepsilon_f=-\bar{\eta}(-1)^{\frac{R_f}{2}}$. In this case, $0< r\leq \frac{R+1}{2}$.

Suppose $H_r$ is an $r$-dimensional subspace of $\F$.
By Lemma~\ref{lem:d_r:2}, we have
\begin{align*}
N(H_r)&=p^{s-r-1}\Big(1+\varepsilon_f \varepsilon_g(p^*)^{-\frac{R-1}{2}}\sum\limits_{(u,v)\in H_r\bigcap (S_f\times S_g)}\overline{\eta}(f(x_u)+g(y_v))\Big)\\
&=p^{s-r-1}\Big(1+p^{-\frac{R-1}{2}}\sum\limits_{(u,v)\in H_r\bigcap (S_f\times S_g)}\overline{\eta}(f(x_u)+g(y_v))\Big).
\end{align*}

For $0< r\leq \frac{R-1}{2}$, by Lemma \ref{lem:2-2} and Lemma \ref{lem:2}, there exists an $r_1$-dimensional subspace
$J_{f,r_1}$ of $ \mathbb{F}_{q_1}$ such that $R_{J_{f,r_1}}=1, \varepsilon_{J_{f,r_1}}=1$ and $J_{f,r_1}\cap \F_{q_1}^{\perp_f}=\{0\}$, and $J_{g,r_2}$ be an $r_2$-dimensional subspace of $J_{e_g}\Big/\F_{q_2}^{\perp_g}$,
 where $1\leq r_1\leq \frac{R_f}{2}, 0\leq r_2\leq \frac{R_g-1}{2}$ and $r_1+r_2=r$ .
Taking \begin{equation}\label{eq:H2} H_{r}=L_f(J_{f,r_1})\times L_g(J_{g,r_2}),\end{equation} then $N(H_r)$ reaches its maximum
$$ N(H_r)=p^{s-(r+1)}\Big(1+(p-1)p^{r-1-\frac{R-1}{2}}\Big) = p^{s-1-r} + (p-1)p^{s-2-\frac{R-1}{2}}. $$
So, for $0< r\leq\frac{R-1}{2}$, the desired result is obtained by Lemma~\ref{lem:d_r:2} and \eqref{eq:d_r:3}.

For $r=\frac{R+1}{2}$, taking $(u,v)\in D$, where $u\in J_{e_f}^{\perp_f}$ and $f(u) \neq 0$.
Define $H_{s-r}=(J_{e_f}\times J_{e_g})\oplus \Big\langle(u,v) \Big\rangle$, which concludes that
$$d_r(C_D)=p^{s-1}-p^{s-r}.$$

\end{proof}

\section{Concluding Remarks}

Inspired by the works of \cite{LYF18,LL22}, using the extended defining-set method, we constructed a family of 3-weight linear codes in \eqref{defcode1}.
Our defining-set is  $ D=\Big\{(x,y)\in \Big(\F_{p^{s_1}}\times\F_{p^{s_2}}\Big)\Big\backslash\{(0,0)\}: f(x)+g(y)=0\Big\}$, $f(x),g(y)$ are $p$-ary quadratic forms over $\mathbb{F}_{p^{s_i}},i=1,2$.
By exponential sum theory together with our previous formula of generalized Hamming weight, we determined their weight distributions and weight hierarchies completely.

 Let $w_{\min}$ and $w_{\max}$ denote the minimum and maximum nonzero weight of our obtained code $C_{D}$ defined in \eqref{defcode1}, respectively.
If $R\geq3$, then it can be easily checked that
$$
 \frac{w_{\min}}{w_{\max}}> \frac{p-1}{p}.
$$
By the results in \cite{AB98} and \cite{21YD06}, we know that every nonzero codeword of $C_{\D}$ is minimal and most of the codes we constructed
are suitable for constructing secret sharing schemes with interesting properties.

\section*{Acknowledgement}

For the research, the first author was supported by the National Science Foundation of China Grant No.12001312
and the second author was supported by Key Projects in Natural Science Research of Anhui Provincial Department of Education No.2022AH050594 and Anhui Provincial Natural Science Foundation No.1908085MA02.

\end{document}